\documentclass{article}
\usepackage{spconf,amsmath,graphicx,hyperref}

\usepackage{setspace}
\usepackage{amsfonts}
\usepackage{epsfig}
\usepackage{subcaption}
\usepackage{mwe}
\usepackage{xcolor}
\usepackage{algorithm}
\usepackage{qcircuit}
\usepackage{braket}
\usepackage{eqnarray}

\usepackage{cite}
\usepackage{amssymb,amsthm}
\usepackage{graphicx}
\usepackage{textcomp}
\usepackage{booktabs}

\usepackage{bm}

\usepackage{array}
\usepackage{algpseudocode}

\title{A Novel Automatic Framework for Speaker Drift Detection in Synthesized Speech}
\name{%
\small $^{1}$Jia-Hong Huang, $^{2}$Seulgi Kim, $^{2}$Yi Chieh Liu, $^{1}$Yixian Shen, $^{1}$Hongyi Zhu, $^{3}$Prayag Tiwari, $^{1}$Stevan Rudinac, $^{1}$Evangelos Kanoulas}
\address{\small $^{1}$University of Amsterdam, $^{2}$Georgia Institute of Technology, $^{3}$Halmstad University}
\begin{document}

\theoremstyle{plain}
\newtheorem{theorem}{Theorem}
\newtheorem{corollary}{Corollary}
\newtheorem{proposition}{Proposition}
\newtheorem{lemma}{Lemma}
\theoremstyle{remark}
\newtheorem*{remark}{Interpretation}

%
\maketitle
\begin{abstract}
Recent diffusion-based text-to-speech (TTS) models achieve high naturalness and expressiveness, yet often suffer from speaker drift, a subtle, gradual shift in perceived speaker identity within a single utterance. This underexplored phenomenon undermines the coherence of synthetic speech, especially in long-form or interactive settings. We introduce the first automatic framework for detecting speaker drift by formulating it as a binary classification task over utterance-level speaker consistency. Our method computes cosine similarity across overlapping segments of synthesized speech and prompts large language models (LLMs) with structured representations to assess drift. We provide theoretical guarantees for cosine-based drift detection and demonstrate that speaker embeddings exhibit meaningful geometric clustering on the unit sphere. To support evaluation, we construct a high-quality synthetic benchmark with human-validated speaker drift annotations. Experiments with multiple state-of-the-art LLMs confirm the viability of this embedding-to-reasoning pipeline. Our work establishes speaker drift as a standalone research problem and bridges geometric signal analysis with LLM-based perceptual reasoning in modern TTS. 
\end{abstract}


\section{\textbf{Introduction}}
Recent progress in text-to-speech (TTS) synthesis, particularly with diffusion-based models, has significantly enhanced the naturalness, expressiveness, and controllability of generated speech \cite{wang2017tacotron, van2016wavenet,meng2024autoregressive,le2023voicebox,shen2023naturalspeech,shen2018natural,huang2024novel,ren2019fastspeech,jeong2021diff,popov2021grad,chen2023schrodinger,huang2025image2text2image,zhang2024towards,lovelace2024sample,lajszczak2024base,huang2025gradient,mehta2024matcha,liu2025e1,he2025continuous}. These models can synthesize long-form utterances with high perceptual fidelity, supporting applications such as personalized virtual assistants, audiobook narration, multi-turn dialog systems, and multimedia systems \cite{huang2023query,zhu2024enhancing,zhu2025interactive,huang2024multi,wang2025reasoning,shen2025macp,liu2018synthesizing,huang2020query,shen2025ssh,huang2021gpt2mvs,di2022dawn,zhang2024beyond,huang2022causal}. However, an underexplored yet critical challenge persists: speaker drift. This phenomenon refers to a subtle, gradual change in the perceived speaker identity within a single utterance, even when the synthesis is conditioned on a single, fixed speaker embedding or prompt. Such intra-utterance inconsistencies can undermine the effectiveness of applications like those above, where maintaining a coherent and stable speaker identity is crucial for a seamless user experience. 

Speaker drift fundamentally differs from conventional speaker changes addressed in diarization or speaker change detection, which typically assume abrupt and discrete transitions \cite{ajmera2004robust,zhang2019fully,xia2022turn}. In contrast, speaker drift involves gradual, often imperceptible shifts in vocal characteristics that accumulate throughout an utterance. This subtle degradation presents significant challenges for detection, quantification, and evaluation. A comparison of related tasks is summarized in Table~\ref{tab:related-tasks}. Currently, there is no standardized evaluation protocol, scalable automated method, or standardized dataset tailored specifically to this problem, leaving a critical gap in quality assurance, model validation, and deployment reliability for both academic research and production-level real-world TTS systems.

To address this gap, we propose a novel, LLM-driven framework for automatic speaker drift detection in synthetic speech. We formulate this as a binary classification task at the utterance level. Specifically, we extract speaker embeddings from short, overlapping segments of each utterance and compute pairwise cosine similarity scores, a compact, interpretable proxy for vocal identity consistency over time. These structured similarity score sequences with specially designed prompts are then fed into state-of-the-art LLMs, e.g., \cite{openai2024gpt4technicalreport,geminiteam2024geminifamilyhighlycapable,anthropic2025claude,deepseekai2025deepseekr1incentivizingreasoningcapability,bai2023qwentechnicalreport}, to assess whether speaker drift is present based solely on the numerical input. This design bypasses the token limitations of modern LLMs, which cannot directly process high-dimensional embeddings. It offers a reference-free, scalable approach that bridges speaker embedding analysis and the reasoning capabilities of LLMs.
 
To support this architecture, we provide theoretical justification for using cosine similarity as an indicator of drift. Under mild distributional assumptions, we prove that a threshold-based classifier using segment-level cosine similarity scores can detect speaker drift with exponentially decreasing error as the similarity gap between same-speaker and different-speaker segments increases. This result formally grounds the use of cosine similarity as a statistically meaningful signal for vocal identity transitions. 

To validate our proposed approach and enable systematic evaluation, we address the lack of real-world data by constructing a controlled benchmark dataset using a high-fidelity, diffusion-based TTS model. Real instances of intra-utterance identity drift are rare, ambiguous, and costly to annotate, making them unsuitable for empirical and systematic study. Instead, we generate utterance samples that either maintain consistent speaker identity or introduce identity shifts within an utterance, then verify them through human annotations. These synthetic samples provide a reliable and reproducible testbed for studying the drift phenomenon and evaluating the proposed framework in a well-defined setting.

Together, our contributions offer a principled pipeline for detecting speaker drift and establish a foundation for future work at the intersection of embedding-based speech analysis and LLM-based perceptual evaluation in modern TTS systems. To our knowledge, this is the first work to (1) define and formalize speaker drift detection as a standalone task, (2) construct a dataset tailored for intra-utterance identity variation, and (3) explore LLM-based reasoning as a diagnostic tool for speaker consistency in TTS pipelines.

\begin{figure*}[t!]
\centering
\scalebox{0.5}{
\centerline{\includegraphics{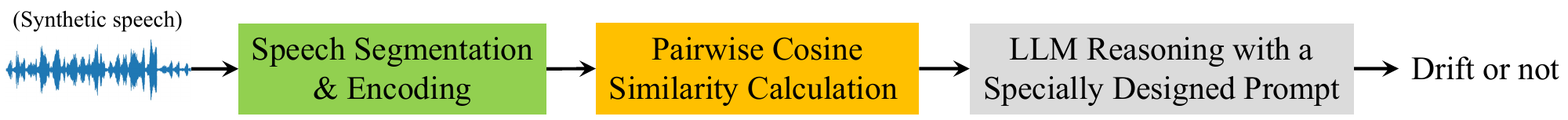}}}
\vspace{-0.3cm}
\caption{Our proposed LLM-based framework for detecting speaker drift in synthesized speech. Further details are provided in Section IV.}
\label{fig1}
\end{figure*}

\begin{table*}[ht]
\centering
\scalebox{1.0}{
\begin{tabular}{|p{4.5cm}|p{4.5cm}|p{4.8cm}|}
\hline
\multicolumn{1}{|c|}{\textbf{Task}} & 
\multicolumn{1}{c|}{\textbf{Goal}} & 
\multicolumn{1}{c|}{\textbf{Key Characteristics}} \\
\hline
Speaker Change Detection & Identify speaker boundaries in natural multi-speaker speech streams & Assumes abrupt, discrete speaker transitions; aims at accurate change-point detection \\
\hline
Speaker Verification or Identification & Confirm or classify speaker identity using labeled reference data & Requires ground-truth speaker labels; evaluates similarity across separate utterances \\
\hline
Speaker Diarization & Segment and cluster speech by speaker in multi-speaker audio & Operates on known or unknown speakers; assumes clear speaker turns; insensitive to subtle intra-speaker drift \\
\hline
Voice Cloning Consistency Evaluation & Evaluate preservation of speaker identity in synthetic speech & Primarily relies on subjective human assessments or limited ABX testing; lacks automated, fine-grained detection methods \\
\hline
Out-of-Distribution Detection (Speaker Embeddings) & Detect anomalous speaker embeddings outside known distribution & Depends on large reference datasets; unsuitable for fine-grained temporal consistency within utterances \\
\hline
Voice Style Transfer Consistency & Assess preservation of non-identity attributes such as prosody, emotion, or accent & Focuses on style or affective features; does not explicitly target speaker identity stability \\
\hline
\textbf{Speaker Drift Detection (Ours)} & \textbf{Detect subtle, gradual variations in speaker identity within a single utterance, particularly in synthetic speech from TTS systems} & \textbf{Reference-free; focuses on temporal embedding consistency within utterances; addresses identity stability challenges unique to TTS} \\
\hline
\end{tabular}}
\vspace{-0.2 cm}
\caption{Comparison of speaker drift detection with related tasks in speaker and TTS research.}
\label{tab:related-tasks}
\end{table*}

\section{\textbf{Dataset Construction}}

\noindent{\textbf{2.1 Overview}}

To systematically study the speaker drift phenomenon in synthetic speech, we construct a benchmark dataset designed for binary classification, determining whether a speaker's identity remains consistent or shifts within a given utterance. 
Reflecting real-world scenarios where speaker consistency is crucial, the dataset features two subtypes for each class: non-drift and hard negative samples for the ``no drift'' class, and abrupt drift and smooth morphing samples for the ``drift'' class. Each sample is created by concatenating consecutive speech segments with a 500 ms silence in between. The dataset contains 32 samples per subtype.

\vspace{+0.1cm}
\noindent{\textbf{2.2 Controlled Synthetic Construction}}

We begin with a curated set of \textup{N} distinct speaker samples. Let \( \mathcal{S} = \{ \mathbf{x}_1, \ldots, \mathbf{x}_{\textup{N}} \} \) denote the set of base speech clips, where each \( \mathbf{x}_i \) is associated with a unique speaker embedding. We synthesize utterances by concatenating three speech segments  $\mathbf{s}_1$, $\mathbf{s}_2$, and $\mathbf{s}_3$:

\begin{itemize}
    \item \textbf{Non-drift (label = 0)}: All segments are selected from the same speaker, e.g., \( [\mathbf{s}_1, \mathbf{s}_1, \mathbf{s}_1] \), producing consistent speaker identity throughout the utterance.
    \item \textbf{Abrupt drift (label = 1)}: At least one segment originates from a different speaker, e.g., \( [\mathbf{s}_1, \mathbf{s}_2, \mathbf{s}_2] \) or \( [\mathbf{s}_1, \mathbf{s}_1, \mathbf{s}_2] \), introducing a discrete speaker shift at a known boundary.
\end{itemize}

\noindent{\textbf{2.3 Hard Negative Construction (label = 0)}}

To evaluate the effectiveness of our proposed LLM-based automatic speaker drift detection framework, we construct hard negative samples, utterances from the same speaker recorded under different conditions that do not involve actual identity drift. These conditions include variations in speaking rate and pitch (e.g., speedup rate = 1.05) and the addition of background noise (e.g., noise level = –30 dB). While these augmentations introduce noticeable acoustic and prosodic changes, they maintain the speaker’s identity. As such, they serve as challenging counterexamples, allowing us to test whether the framework can reliably distinguish true speaker drift from superficial style or environmental variations.

\vspace{+0.1cm}
\noindent{\textbf{2.4 Modeling Gradual or Subtle Drift (label = 1)}}

To simulate gradual speaker identity transitions in a more realistic setting, we introduce smooth speaker morphing, where speaker drift is performed directly at the audio level. Instead of generating discrete segments from different speakers, we synthesize overlapping speech regions and apply time-domain blending to produce perceptually smooth identity changes.
Let \( \mathbf{x}_A(t) \) and \( \mathbf{x}_B(t) \) denote two different speech waveforms. A morphing region \( t \in [T_1, T_2] \subset [0, T] \) is defined over which the audio is blended using a linear cross-fade:
$\mathbf{x}_{\text{morph}}(t) = (1 - \alpha(t)) \cdot \mathbf{x}_A(t) + \alpha(t) \cdot \mathbf{x}_B(t)$, where $\alpha(t) = \frac{t - T_1}{T_2 - T_1}$.
\noindent Outside the morphing region, the waveform is taken entirely from one speaker:
\[
\mathbf{x}(t) =
\begin{cases}
\mathbf{x}_A(t), & t < T_1 ,\\
\mathbf{x}_{\text{morph}}(t), & T_1 \leq t \leq T_2 ,\\
\mathbf{x}_B(t), & t > T_2.
\end{cases}
\]
\noindent This results in a continuous utterance where the speaker identity shifts gradually from \( A \) to \( B \) over a defined time window (e.g., from 3s to 6s in Figure \ref{fig:speaker_morphing}), mimicking realistic speaker drift at the acoustic level.
Compared to hard cuts, this approach produces more subtle transitions and challenges models to detect non-abrupt identity shifts that lack clear boundaries. It serves as an essential component for evaluating model sensitivity to fine-grained speaker variation.

Although synthetic, our proposed dataset exhibits natural speech-like quality and variability, closely resembling real human voices. Each sample has been manually inspected to ensure clarity, coherence, and intended speaker characteristics, guaranteeing high perceptual quality. The dataset allows precise control over the timing and location of speaker drift within utterances, enabling fine-grained analysis. It is label-balanced, containing equal numbers of drifted and non-drifted samples to avoid classification bias. Furthermore, utterances are structured in fixed-length segments, facilitating segment-level inspection and simplifying downstream detection.

\begin{figure}[t!]
    \centering
    \includegraphics[width=1.0\linewidth]{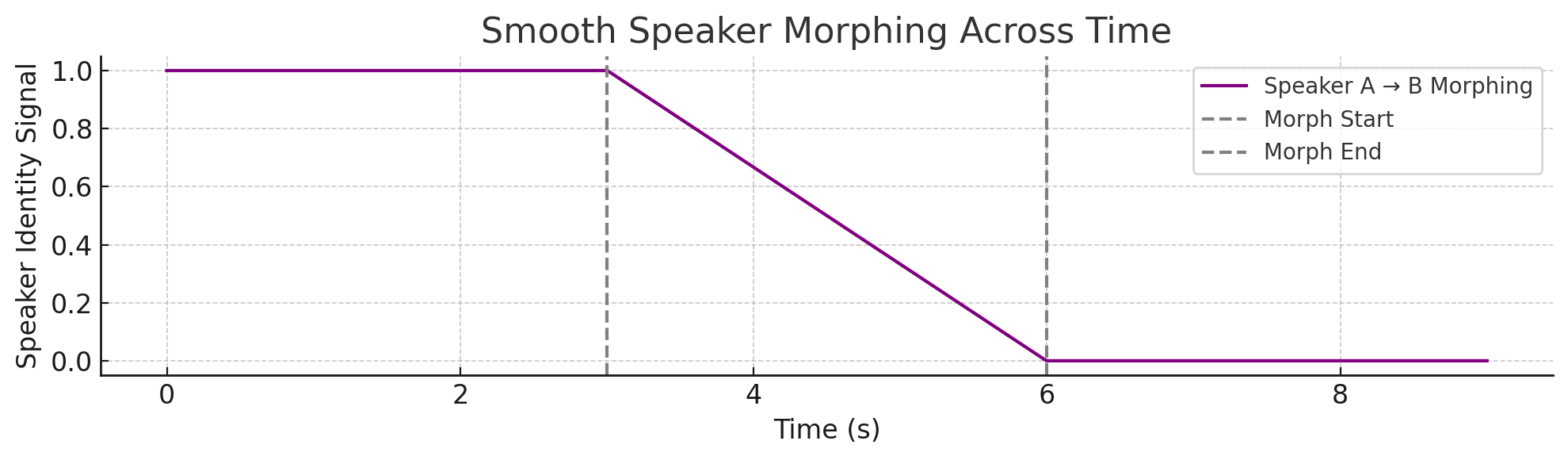}
    \vspace{-0.9cm}\caption{Smooth speaker morphing across an utterance. The morphing region (3s–6s) involves audio-level cross-fading from Speaker A to Speaker B.}
    \label{fig:speaker_morphing}
\end{figure}

\section{\textbf{Method}}

\vspace{+0.1 cm}
\noindent{\textbf{3.1 Problem Formulation}}

Given a speech utterance \( \mathbf{x}(t) \), divided into three contiguous segments \( (s_1, s_2, s_3) \) of equal duration, the goal is to detect whether the speaker identity remains consistent throughout or experiences a drift at one or more boundaries.
We formalize the task as a binary classification problem. Let the label \( y \in \{0, 1\} \), where \( y = 0 \) indicates a consistent speaker and \( y = 1 \) indicates at least one speaker shift. Our proposed method aims to predict \( y \) from the acoustic features or embeddings derived from \( \mathbf{x}(t) \).

\noindent{\textbf{3.2 Speaker Embedding Extraction and Cosine Similarity}}

To enable automatic detection of speaker drift, each utterance is divided into three consecutive audio segments $(\mathbf{s}_1, \mathbf{s}_2, \mathbf{s}_3)$, from which we extract fixed-dimensional speaker embeddings $\mathbf{e}_i = f_{\text{embed}}(\mathbf{s}_i) \in \mathbb{R}^d$ using a pre-trained model (e.g., Wav2Vec2). We then compute cosine similarities between adjacent segments to quantify inter-segment identity consistency: $\text{sim}_{i,j} = \cos(\mathbf{e}_i, \mathbf{e}_j) = \frac{\mathbf{e}_i^\top \mathbf{e}_j}{\|\mathbf{e}_i\| \cdot \|\mathbf{e}_j\|}$, where $ (i,j) \in \{(1,2), (2,3)\}.$
The resulting similarity pair $(\text{sim}_{1,2}, \text{sim}_{2,3})$ serves as a compact representation of speaker consistency across the utterance, where lower values indicate potential identity shifts and provide an interpretable signal for LLM-based inference.

\vspace{+0.1 cm}
\noindent{\textbf{3.3 Prompt Design}}

We design a standardized prompt template to ensure consistent and interpretable LLM-based evaluation of speaker drift. Each prompt concisely conveys the task objective, detecting speaker identity shifts, alongside the relevant pairwise cosine similarity scores. By emphasizing clarity and minimizing ambiguity, the prompt guides the LLM toward accurate and reproducible decisions. The structure is as follows: \\
\textbf{Instruction template:}
\begin{quote}
You are given a list of $N$ pairwise similarity scores derived from three consecutive segments of each utterance. Each score reflects the similarity between adjacent segments, computed using \{\texttt{your similarity metric}\}, with embeddings obtained from \{\texttt{your audio encoder}\}.
Your task is to assess whether each utterance exhibits speaker drift based on the provided similarity scores. For instance, consider the example: \texttt{(0.9963, 0.9872)}. Based on this, determine whether a speaker identity shift has likely occurred. Provide a binary decision, \texttt{same (no drift)} or \texttt{different (drift)}, and briefly explain your reasoning.
Evaluate each case independently and report your decisions with corresponding justifications.
\end{quote}

\noindent{\textbf{3.4 Theoretical Justification for Drift Detection via Embedding Similarity}}

Detecting speaker drift hinges on a key geometric intuition: embeddings from the same speaker tend to cluster together, while those from different speakers exhibit separation, particularly when measured via cosine similarity. This section formalizes that intuition through two complementary results: one intuitive and qualitative, and one rigorous and quantitative.

\vspace{+0.1 cm}
\noindent{\textbf{3.4.1 Motivating Insight of Drift Separability under Embedding Smoothness}}

We begin with an intuitive proposition that connects smooth embedding behavior with the ability to detect drift using cosine similarity.

\begin{proposition}
[\textbf{Embedding Smoothness \& Drift Separability}]
Let $\mathbf{e}_1, \mathbf{e}_2, \mathbf{e}_3 \in \mathbb{R}^d$ be embeddings from three contiguous speech segments. Define pairwise cosine similarities $\text{sim}_{i,j} = \cos(\mathbf{e}_i, \mathbf{e}_j)$. Suppose:
$\max(\text{sim}_{1,2}, \text{sim}_{2,3}) < \tau$ for some threshold $\tau \in (0,1)$. Then any classifier $f$ that is Lipschitz-continuous over the similarity space can separate drift from non-drift samples with bounded error:
$\mathbb{P}\left(f(\text{sim}_{1,2}, \text{sim}_{2,3}) \neq y\right) \leq \epsilon(\tau)$, where $\epsilon(\tau) \to 0$ as $\tau \to 0$, assuming embeddings vary smoothly for the same speaker and show sufficient separation across speakers.
\end{proposition}

This proposition motivates the use of cosine similarity as a natural proxy for speaker identity consistency. It suggests that if embeddings change gradually within the same speaker and shift abruptly across speakers, then even simple decision boundaries (e.g., thresholding) can detect drift reliably.

\begin{table*}[t!]
\centering
\scalebox{0.6}{
\begin{tabular}{cccccccc}
\toprule
\textbf{Embedding Type} & \textbf{GPT-4o} \cite{openai2024gpt4technicalreport} & \textbf{Gemini-Pro-2.5} \cite{geminiteam2024geminifamilyhighlycapable} & \textbf{Claude-4} \cite{anthropic2025claude} & \textbf{Qwen-3} \cite{bai2023qwentechnicalreport} & \textbf{DeepSeek-R1} \cite{deepseekai2025deepseekr1incentivizingreasoningcapability} & \textbf{PCA-based Baseline} & \textbf{Fixed-threshold Baseline}\\
\midrule
Wav2Vec2 \cite{baevski2020wav2vec} & \textbf{90.70\%} & 82.90\% & 88.20\% & 72.70\% & 80.00\% & 71.30\% & 61.70\% \\
MFCC \cite{davis1980comparison}    & 39.00\% & 38.64\% & 76.00\% & 71.40\% & \textbf{80.00\%} & 65.60\% &  57.40\% \\
Whisper \cite{radford2023robust}   & 89.41\% & 80.00\% & 84.38\% & 72.70\% & \textbf{90.60\%} & 67.30\% &  61.30\% \\
\midrule
\textbf{Thresholds} & (0.960, 0.950, 0.995) & (0.970, 0.950, 0.998) & (0.950, 0.950, 0.995) & (0.910, 0.950, 0.990) & (0.970, 0.990, 0.995) & (0.950, 0.950, 0.950) & (0.900, 0.900, 0.900) \\
\bottomrule
\end{tabular}}
\vspace{-0.2cm}
\caption{Performance comparison with baselines and ablation studies using different audio embedding methods, evaluated using F1 score and pairwise cosine similarity.}
\label{tab:model_comparison_drift}
\end{table*}

\vspace{+0.1 cm}
\noindent{\textbf{3.4.2 Core Theoretical Result: Error Bound under Distributional Assumptions}}

We now present a formal result under distributional assumptions on cosine similarity scores for drift vs. non-drift segments.
Before stating the theorem, note that we assume the embeddings lie on the unit sphere $\mathbb{S}^{d-1} \subset \mathbb{R}^d$, where $\mathbb{S}^{d-1} = \{ \mathbf{x} \in \mathbb{R}^d : \|\mathbf{x}\| = 1 \}$. This reflects a common normalization step in modern speaker embedding systems (e.g., x-vectors, Wav2Vec2), where embeddings are constrained to have unit norm to make cosine similarity equivalent to the dot product. Since the unit sphere is a $(d-1)$-dimensional manifold, each embedding lies on a curved surface, not a flat space, with exactly one degree of freedom removed due to the unit-norm constraint.

\begin{theorem}
[\textbf{Embedding Separation Bound for Speaker Drift Detection}]
Let $\mathbf{e}_1, \mathbf{e}_2, \mathbf{e}_3 \in \mathbb{S}^{d-1}$ denote unit-norm embeddings corresponding to three contiguous speech segments. We assume that for non-drift (i.e., same-speaker) samples, the expected pairwise cosine similarity satisfies $\mathbb{E}[\cos(\mathbf{e}_i, \mathbf{e}_j)] \geq \mu_0$ with variance bounded by $\sigma^2$. In contrast, for drift (i.e., different-speaker) samples, the expected similarity is lower, satisfying $\mathbb{E}[\cos(\mathbf{e}_i, \mathbf{e}_j)] \leq \mu' < \mu_0$, with the same variance $\sigma^2$.
Define the classifier:
\[
f(\mathbf{e}_1, \mathbf{e}_2, \mathbf{e}_3) =
\begin{cases}
 1 & \text{if } \min\left(\cos(\mathbf{e}_1, \mathbf{e}_2), \cos(\mathbf{e}_2, \mathbf{e}_3)\right) < \tau, \\
 0 & \text{otherwise,}
 \end{cases}
\]
for any threshold $\tau \in (\mu', \mu_0)$. Then, the misclassification probability is bounded by:
\begin{align*}
\mathbb{P}\left(f(\mathbf{e}_1, \mathbf{e}_2, \mathbf{e}_3) \neq y\right)
&\leq 4 \exp\left(-\frac{\Delta^2}{2\sigma^2}\right),
\quad \text{where} \\ \Delta &= \min(\mu_0 - \tau, \tau - \mu').
\end{align*}
\end{theorem}

This result provides a concrete error bound, showing that the classifier’s performance improves exponentially as the separation margin $\Delta$ grows. It justifies cosine thresholding as a statistically grounded and computationally simple strategy for detecting speaker drift, provided embeddings are well-separated across speakers and stable within speakers.

\begin{proof}[\textless \textbf{Proof}\textgreater]
We aim to bound the misclassification probability of the classifier
$f(\mathbf{e}_1, \mathbf{e}_2, \mathbf{e}_3)$.
Let $y = 1$ denote a drift case (i.e., a speaker change occurs), and $y = 0$ denote a non-drift case (i.e., all segments are from the same speaker).
Let $\text{sim}_{i,j} := \cos(\mathbf{e}_i, \mathbf{e}_j)$ denote the cosine similarity between embeddings $\mathbf{e}_i$ and $\mathbf{e}_j$. In the non-drift case, we assume $\mathbb{E}[\text{sim}_{1,2}] \geq \mu_0$ and $\mathbb{E}[\text{sim}_{2,3}] \geq \mu_0$, with variance $\text{Var}(\text{sim}_{i,j}) \leq \sigma^2$. In the drift case, the expected similarities are lower: $\mathbb{E}[\text{sim}_{1,2}] \leq \mu'$, $\mathbb{E}[\text{sim}_{2,3}] \leq \mu' < \mu_0$, with the same variance bound $\sigma^2$. Let the classifier threshold $\tau$ satisfy $\mu' < \tau < \mu_0$, and define the decision margin as $\Delta = \min(\mu_0 - \tau, \tau - \mu')$.

\noindent\textbf{Step 1: Bounding False Positives (Type I Error)}

Suppose the input is non-drift. A false positive occurs when:
$
f(\mathbf{e}_1, \mathbf{e}_2, \mathbf{e}_3) = 1 \quad \text{(i.e.,} \min(\text{sim}_{1,2}, \text{sim}_{2,3}) < \tau \text{)}.
$
This implies at least one of the similarities is below $\tau$, so:
$
\mathbb{P}(f \neq y \mid y = 0) \leq \mathbb{P}(\text{sim}_{1,2} < \tau) + \mathbb{P}(\text{sim}_{2,3} < \tau).
$

\noindent Since $\mathbb{E}[\text{sim}_{i,j}] \geq \mu_0$, and $\tau < \mu_0$, we apply Hoeffding’s inequality (for bounded variables, e.g., cosine similarity in $[-1, 1]$):
$
\mathbb{P}(\text{sim}_{i,j} < \tau) \leq \exp\left(-\frac{(\mu_0 - \tau)^2}{2\sigma^2}\right).
$
Thus,
$
\mathbb{P}(f \neq y \mid y = 0) \leq 2 \exp\left(-\frac{(\mu_0 - \tau)^2}{2\sigma^2}\right).
$

\noindent\textbf{Step 2: Bounding False Negatives (Type II Error)}

Now suppose the input is drift. A false negative occurs when:
$
f(\mathbf{e}_1, \mathbf{e}_2, \mathbf{e}_3) = 0 \quad \text{(i.e.,} \min(\text{sim}_{1,2}, \text{sim}_{2,3}) \geq \tau \text{)}.
$
That is, both similarities are above $\tau$:
$
\mathbb{P}(f \neq y \mid y = 1) \leq \mathbb{P}(\text{sim}_{1,2} \geq \tau) + \mathbb{P}(\text{sim}_{2,3} \geq \tau).
$
Since $\mathbb{E}[\text{sim}_{i,j}] \leq \mu'$, and $\tau > \mu'$, again apply Hoeffding’s inequality:
$
\mathbb{P}(\text{sim}_{i,j} \geq \tau) \leq \exp\left(-\frac{(\tau - \mu')^2}{2\sigma^2}\right).
$
So,
$
\mathbb{P}(f \neq y \mid y = 1) \leq 2 \exp\left(-\frac{(\tau - \mu')^2}{2\sigma^2}\right).
$

\noindent\textbf{Final Bound: Total Classification Error}
 
Combining both cases:
$
\mathbb{P}(f \neq y) \leq 2 \exp\left(-\frac{(\mu_0 - \tau)^2}{2\sigma^2}\right) + 2 \exp\left(-\frac{(\tau - \mu')^2}{2\sigma^2}\right).
$
By definition of $\Delta = \min(\mu_0 - \tau, \tau - \mu')$, we have:
$
\mathbb{P}(f \neq y) \leq 4 \exp\left(-\frac{\Delta^2}{2\sigma^2}\right). 
$
\end{proof}

\section{\textbf{Experiments}}

\noindent{\textbf{4.1 Experimental Setup}}

\noindent\textbf{Dataset and Evaluation.}
We use the dataset described in Section III, consisting of 128 samples, 64 with speaker drift and 64 without, synthesized from 384 high-quality utterances by different speakers and verified by human annotators. Each sample is 9 to 40 seconds long. For evaluation, we report accuracy and F1 score. Each sample is represented by (1) cosine similarity scores between adjacent segments and (2) principal component analysis (PCA)-reduced speaker embeddings, where each segment is compressed to 8 or 16 dimensions, yielding 24- or 48-dimensional inputs. LLMs, including GPT-4o \cite{openai2024gpt4technicalreport}, Gemini-Pro 2.5 \cite{geminiteam2024geminifamilyhighlycapable}, Claude-4 \cite{anthropic2025claude}, DeepSeek-R1 \cite{deepseekai2025deepseekr1incentivizingreasoningcapability}, and Qwen-3 \cite{bai2023qwentechnicalreport}, are tested under both zero-shot and few-shot settings on the full dataset.

\noindent\textbf{Baselines.}
For the fixed-threshold baseline, we classify a sample as drift if either $\cos(\mathbf{s}_1, \mathbf{s}_2)$ or $\cos(\mathbf{s}_2, \mathbf{s}_3)$ falls below 0.90; otherwise, it is labeled non-drift. This threshold, selected from the empirical distribution of minimum similarity scores, balances overlap between classes; non-drift samples occasionally drop to 0.77, while drift samples can reach as low as 0.80. Despite this compromise, the method achieves a modest F1 score of 0.62, highlighting the limitations of such simple heuristics. To build a strong baseline, we project segment-level speaker embeddings into a lower-dimensional space via PCA to preserve broader variation patterns. These reduced embeddings are then used as input to LLMs for binary classification.

\vspace{+0.1 cm}
\noindent{\textbf{4.2 Experimental Results}}

\noindent\textbf{Performance Analysis.}
We evaluate our LLM-based method against the fixed-threshold and PCA-based baselines described above. As reported in Table \ref{tab:model_comparison_drift}, our approach yields a substantially higher F1 score than the fixed-threshold baseline and consistently outperforms the PCA-based baseline, indicating superior performance in detecting speaker drift. These results validate the effectiveness of leveraging LLMs with pairwise cosine similarity scores for the speaker drift detection task.

\noindent\textbf{Ablation Studies.}
We conduct several ablation studies to assess the impact of different design choices:
(1) \textit{Audio Embedding Type:} We compare Wav2Vec2, MFCC, and Whisper embeddings as input features. As shown in Table \ref{tab:model_comparison_drift}, Wav2Vec2 embeddings achieve the highest performance, indicating their superior ability to capture speaker-relevant characteristics.
(2) \textit{PCA Dimensionality:} We assess the effect of dimensionality reduction by compressing embeddings to 8 dimensions (preserving $\approx 75$\% of the variance) and 16 dimensions ($\approx 87$\% variance preserved) per segment. Results in Table \ref{tab:drift-results} show that the 16-dimensional setting yields better performance, suggesting that retaining more dimensions helps preserve discriminative speaker information.
(3) \textit{Input Representation:} We compare models using pairwise cosine similarity scores against those using PCA-reduced embeddings as input. As shown in Table \ref{tab:drift-results}, the former consistently outperforms the latter, indicating that explicit relational features between segments more effectively capture speaker identity drift than compressed embedding representations.

\begin{table}[t!]
\centering
\scalebox{0.65}{
\begin{tabular}{c|c|c|c}
\toprule
\textbf{Method} & \textbf{Input Type} & \textbf{Accuracy} & \textbf{F1 Score} \\
\midrule

GPT-4o         & PCA Embeddings (8) & 50.3\% & 66.7\% \\
GPT-4o         & PCA Embeddings (16) & 73.4\% & 73.6\% \\
GPT-4o         & \textbf{Cosine Scores}  & \textbf{89.5\%} & \textbf{90.7\%} \\
\midrule
Gemini-Pro-2.5 & PCA Embeddings (8)  & 50.8\% & 58.3\% \\
Gemini-Pro-2.5 & PCA Embeddings (16)  & 52.3\% & 59.8\% \\
Gemini-Pro-2.5 & \textbf{Cosine Scores}  & \textbf{79.7\%} & \textbf{82.9\%} \\
\midrule
Claude-4   & PCA Embeddings (8)  & 63.7\% & 69.3\% \\
Claude-4   & PCA Embeddings (16)  & 67.5\% & 73.6\% \\
Claude-4   & \textbf{Cosine Scores}  & \textbf{83.4\%} & \textbf{88.2\%} \\
\midrule
Qwen-3       & PCA Embeddings (8) & 65.3\% & 68.0\% \\
Qwen-3       & PCA Embeddings (16) & 69.1\% & 72.4\% \\
Qwen-3       & \textbf{Cosine Scores} & \textbf{69.5\%} & \textbf{72.7\%} \\
\midrule
DeepSeek-R1    & PCA Embeddings (8) & 60.9\% & 71.2\% \\
DeepSeek-R1    & PCA Embeddings (16) & 63.4\% & 74.4\% \\
DeepSeek-R1    & \textbf{Cosine Scores}  & \textbf{78.8\%} & \textbf{80.0\%} \\
\bottomrule
\end{tabular}}
\vspace{-0.2 cm}
\caption{Ablation study of the proposed LLM-driven speaker drift detection framework, evaluating the impact of different LLM backbones and input feature formats.}
\label{tab:drift-results}
\end{table}

\section{\textbf{Conclusion and Future Work}}
In this work, we introduced the first automated framework for detecting speaker drift in diffusion-based TTS, leveraging cosine similarity as a theoretically grounded proxy for vocal identity consistency and prompting LLMs for perceptual reasoning. Our method bridges low-level acoustic embeddings with high-level evaluation and is supported by a new benchmark dataset with human-verified annotations. Looking forward, we plan to extend this framework to multilingual and cross-lingual settings and explore fine-tuning LLMs for even greater sensitivity to subtle prosodic and identity cues in generated speech.

\begin{spacing}{0.6}
\bibliographystyle{IEEEbib}
\bibliography{strings,refs}
\end{spacing}

\end{document}